\documentclass[twoside,a4paper,11pt]{article}

\pdfoutput=1

\usepackage[english]{babel}
\usepackage{amsmath}
\usepackage{amssymb}
\usepackage{amsthm}
\usepackage{mathrsfs}
\usepackage{graphicx}
\usepackage{hyperref}

\newcommand{\n}{\noindent}
\newcommand{\be}{\begin{equation*}}
\newcommand{\ee}{\end{equation*}}
\renewcommand{\Im}{\operatorname{Im}\,}
\renewcommand{\Re}{\operatorname{Re}\,}

\newcommand{\Hbessel}{\operatorname{H}}

\newcommand{\RE}{\mathbb R}
\newcommand{\CO}{\mathbb C}

\newcommand{\OO}{\mathcal O}
\newcommand{\ve}{\varepsilon}
\newcommand{\ga}{\gamma}

\newcommand{\al}{\alpha}
\newcommand{\la}{\lambda}

\newcommand{\hh}{H}
\newcommand{\rr}{R}

\newcommand{\HH}{\mathscr{H}}

\newcommand{\BB}{\mathcal B}

\newcommand{\zz}[1]{z^{(#1)}}
\newcommand{\lla}[1]{\la^{(#1)}}
\newcommand{\ww}[1]{w^{(#1)}}
\newtheorem{theorem}{Theorem}
\newtheorem{proposition}{Proposition}

\theoremstyle{definition}

\newtheorem{definition}{Definition}
\newtheorem{remark}{Remark}

\theoremstyle{remark}

\numberwithin{equation}{section}

\title{Perturbations of eigenvalues embedded at threshold: one, two and three dimensional solvable models}
\author{Claudio Cacciapuoti${}^1$, Raffaele Carlone${}^2$, and Rodolfo Figari${}^3$
\\
\\
${}^1$Hausdorff Center for Mathematics\\
Institut f\"ur Angewandte Mathematik, Bonn Universit\"at  \\
Endenicher Allee 60, 53115 Bonn, Germany\\
cacciapuoti@him.uni-bonn.de\\ \\
${}^2$Dipartimento di 	Fisica e Matematica, Universit\`a degli Studi Insubria\\
Via Valleggio 11, 22100 Como, Italy\\
raffaele.carlone@me.com\\ \\
${}^3$Istituto Nazionale di Fisica Nucleare (INFN), Sezione di Napoli\\
Dipartimento di Scienze Fisiche, Universit\`a di Napoli Federico II\\
Via Cintia 80126 Napoli, Italy\\
figari@na.infn.it.
}
\date{}
\begin{document}
\maketitle

\vspace{0.5cm}

\hspace{0,5cm}\emph{ In memory of Pierre Duclos}

\begin{abstract}
We examine perturbations of eigenvalues and resonances for a class of multi-channel quantum mechanical model-Hamiltonians describing a particle interacting with a localized spin  in dimension $d=1,2,3$.  We consider unperturbed Hamiltonians showing eigenvalues and resonances at the threshold of the continuous spectrum and we analyze the effect of various type of perturbations on the spectral singularities. We provide algorithms to obtain convergent series expansions for the coordinates of the singularities.  
\end{abstract}

\begin{small}
\n
\end{small}

%%%%%%%%%%%%%%%%%%%%%%%%%%%%%%%%%%%%%%%
%%%%%%%%%%%%%%%%%%%%%%%%%%%%%%%%%%%%%%%
%SECTION
%%%%%%%%%%%%%%%%%%%%%%%%%%%%%%%%%%%%%%%
%%%%%%%%%%%%%%%%%%%%%%%%%%%%%%%%%%%%%%%

\section{Introduction}
%--------------------introduzione al tema delle risonanze a soglia-----------------------%

An extensive recent literature presenting different rigorous approaches  to the analysis of perturbations of  energy eigenvalues embedded in the continuous part of the spectrum of Schr\"{o}dinger operators  is now available (see e.g. \cite{DES95},  \cite{DEM01}, \cite{Exn91}, \cite{Hun90}, \cite{Kin91}, \cite{SW98} and references therein).

%--------------------riassunto dell`articolo di Jensen e Nenciu---------------------------%
Less is known about the case of eigenvalues embedded at the threshold of the continuous spectrum. Mainly because dilation-analyticity techniques loose their effectiveness when applied to the study of perturbations of bound states or resonances at a threshold, results on this particular case are rare. At the best of our knowledge the only recent work on this topic was done by A. Jensen and G. Nenciu \cite{jensen-nenciu:06}.

\noindent They consider the case of a Schr\"{o}dinger operator
\[
H(\ve)=-\Delta+V+\ve W\equiv H+\ve W
\]
on $L^{2}(\mathbb{R}^{n})$ with $n$ odd. The unperturbed Hamiltonian $H$ is assumed to have a non degenerate eigenvalue at zero energy and suitable hypotheses on $V$ guarantee that the essential spectrum of $H$ is purely absolutely continuous and fills the half line $[0,\infty )$. The self-adjoint operator $W$ is assumed to have strictly positive expectation value in the eigenvector of $H$ at zero energy, in such a way that the singularity is ``pushed up" by the perturbation. 

\noindent Under some technical assumptions on the properties of $(H-z)^{-1}$ for complex $z$, close to the origin, the authors  prove that, as effect of the perturbation, the zero eigenvalue develops into a resonance. 

\noindent  It is known that the perturbation drives an eigenvalue strictly embedded in the continuum spectrum of $H$ to move into a singularity of the resolvent in a non-physical Riemann sheet. The  imaginary part of the singular point position grows in that case like $\ve^{2}$ for small $\ve$.  On the other end, in the case of an  eigenvalue at threshold, Jensen and Nenciu find a behavior like $\ve^{2+\nu/2}$, with $\nu$ integer $\nu\geqslant -1$. This implies that  the lifetime of the corresponding resonances shows a universal dependence on the perturbation strength in the former case whereas in the threshold case it can be larger for $\nu \geqslant 1$, or smaller in the case $\nu=-1$. In several examples of one and two channel Schr\"odinger operators in dimension 1 and 3 the authors detail the computation of the leading order in $\ve$ of the lifetime of the resonance. They show explicit cases where $\nu \geqslant 1$ and others where  $\nu=-1$.

\noindent In the following, in order to investigate formation of resonances by perturbation of threshold eigenvalues we will make use of Hamiltonians  characterizing different dynamical models for a quantum particle moving in an array of localized spins (see \cite{cacciapuoti-carlone-figari:07} for details). By changing their geometrical and dynamical parameters the spectral structure of the Hamiltonians can be adapted to show isolated or embedded eigenvalues as well as eigenvalues and resonances at any threshold of the continuous spectrum. In this paper we analyze only two channel Hamiltonians for one particle and one spin with eigenvalues and/or resonances in the upper channel at the continuum threshold. 

\noindent Inasmuch as we examine specific models our analysis lacks some generality. However we want to point out that zero-range interaction Hamiltonians are particularly versatile models reproducing all the qualitative dynamical features typical of short range potentials (see e.g. \cite{aghh:05}, \cite{DO88}). In particular one and many channel point interaction Hamiltonians share the property of having a resolvent expansion around the origin of the type assumed by Jensen and Nenciu to prove their results. Conversely any Schr\"odinger operator having such resolvent expansion around the origin can be approximated by point interaction Hamiltonians. Moreover the high degree of computability typical of the models we discuss here allows us to consider also two dimensional cases and write down explicitly convergent series expansions for the coordinates of resonances and eigenvalues. 

\noindent
In this paper we will only consider multichannel Hamiltonians in order to have a more direct comparison with the results obtained for embedded eigenvalues in \cite{cacciapuoti-carlone-figari:09}. One channel Hamiltonians, possibly with multiple point-scatterers, also show very rich spectral configurations with resonances and eigenvalues at some continuous threshold. We plan to examine those cases in further work.

\noindent
We will not state here our results in a time-dependent framework, limiting ourselves to the so called spectral form of the Fermi Golden Rule. The investigation of the ``survival probability" of the resonant state is a fundamental step in order to investigate reality and time range of validity of  the expected exponential behavior  on which the very notion of lifetime relies. We mention that the complete knowledge of the Hamiltonian generalized eigenfunctions allows in our case a very detailed analysis of the time evolution of the resonant state as it was done in the case of embedded eigenvalues  (\cite{cacciapuoti-carlone-figari:09}). 

\noindent
The paper is organized as follows. In section \ref{sec1} we introduce notation and basic definitions. In section \ref{sec2} we state and prove our results. Within this section we split in subsections  the analysis of different kinds of perturbations. A final section consists of  a summary of results together with further comments.

%%%%%%%%%%%%%%%%%%%%%%%%%%%%%%%%%%%%%%%
%%%%%%%%%%%%%%%%%%%%%%%%%%%%%%%%%%%%%%%
%SECTION
%%%%%%%%%%%%%%%%%%%%%%%%%%%%%%%%%%%%%%%
%%%%%%%%%%%%%%%%%%%%%%%%%%%%%%%%%%%%%%%
\section{
\label{sec1}
Basic definitions and results}

For $d=1,2,3,$ we consider  the Hilbert space  $\HH:=L^2(\RE^d)\oplus
L^2(\RE^d)$. We  denote by $\Psi$ the generic (column) vector in $\HH$:
\[
\Psi=\begin{pmatrix}
 \psi_0\\
 \psi_1
 \end{pmatrix}\qquad \psi_j\in L^2(\RE^d)\quad j=0,1\,.
\]
 $\HH$ is the state space of a quantum particle in $\RE^d$ in presence of
a two-level quantum system (a spin) localized at the origin.  $\psi_0$ (
resp. $\psi_1$) represents  the state of the particle in the channel
where the spin (better a particular component of the vector spin
operator) has value $-1$ (resp. $+1$).  In what follows we consider 
Hamiltonians in $\HH$  belonging to the family of the self-adjoint
extensions of the symmetric operator $S$ defined by
$D(S):=C_0^\infty(\RE^d\backslash\{0\})\oplus
C_0^\infty(\RE^d\backslash\{0\})$, $S\Psi:=(-\Delta\psi_0,
(-\Delta+1)\psi_1)$. We do not detail here  how to characterize the whole
family of self-adjoint extensions of $S$, as this was done, in a more
general setting and with slightly different notation, elsewhere (see
\cite{cacciapuoti-carlone-figari:07} for $d=1,3$, and  \cite{cacciapuoti-carlone-figari:09} for $d=2$). According to
the definition of $S$ the 0-channel (resp. the 1-channel) will be
referred to as the lower (resp. the upper) channel.

For $z\in\CO\backslash\RE^+$, we denote with $G^z(x)$ the fundamental
solution  in $\RE^d$ of Helmoltz's equation:  $(-\Delta-z)G^z=\delta$;
explicitly
\begin{equation}
\label{Gz}
G^z(x)=\left\{\begin{aligned}
&i\frac{e^{i\sqrt{z}|x|}}{2\sqrt{z}}
& d=1\\
&\frac{i}{4}\Hbessel^{(1)}_0\big(\sqrt{z}\,|x|\big)\qquad& d=2
\\
&\frac{e^{i\sqrt{z}|x|}}{4\pi|x|}& d=3
\end{aligned}\right.
\qquad\textrm{with}\quad z\in\CO\backslash\RE^+\,;\;\Im(\sqrt{z})>0
\end{equation}
Here $\Hbessel^{(1)}_0\big(\eta\big)$ is the zero-th Bessel function of third
kind (also known as Hankel function of the first kind), see, e.g., \cite{abramowitz-stegun:72}. We recall that
$\Hbessel^{(1)}_0\big(\eta\big)$ tends to zero
as $|\eta|\to \infty$ for $\Im\eta>0$ and that it has a logarithmic
singularity in zero
\begin{equation*}
\Hbessel^{(1)}_0\big(\eta\big)=
\frac{2i}{\pi}\ln\frac{\eta}{2}+1+\frac{2i\ga}{\pi}+\OO(\ln(\eta)\eta^2)\,,
\end{equation*}
where $\ga$ is the Euler's constant $(\ga\simeq 0.577)$.

Notice that, for $\Im(\sqrt{z})>0$, $G^z(x) \in L^2(\RE^d)$ for $d=1,2,3$,
a property which does not hold in higher dimensions. This is a crucial
feature in the definition of our model-Hamiltonians and it is the reason
why point interaction Hamiltonians are trivial in dimensions bigger than
three.

We denote by $\hh_0$ the Hamiltonian in $\HH$ given in the following
%%%%%%%%%
%DEFINITION
%%%%%%%%%
\begin{definition}Let $\theta_0\in\RE$ and
%%%
\begin{equation}
\label{theta1}
\theta_1=
\left\{\begin{aligned}
&2&& d=1\\
&a&& d=2\\
&-4\pi&\quad& d=3
\end{aligned}
\right.
\qquad a=\frac{2\pi}{\ln(2)-\gamma}\simeq 54\,.
\end{equation}
%%%
$\hh_{0}:D(\hh_{0})\subset \HH\to\HH$ is the self-adjoint operator:
%%%
\be
%\label{DH0}
\begin{aligned}
D(\hh_0):=\bigg\{\Psi\equiv
\begin{pmatrix}
\psi_{0}\\
\psi_{1}
\end{pmatrix}
\in\HH\bigg|&\psi_0=\phi_0^z+q_0 G^{z}\,,\;\psi_1=\phi_1^z+q_1 G^{z-1}\,;
\phi_0^z\,,\;\phi_1^z\in H^2(\RE^d)\,;\;z\in\CO\backslash\RE\,;\\
&q_0=\theta_0f_0\,,\;q_1=\theta_1 f_1\,;\\
&
\begin{aligned}
&f_j=\psi_j(0)&&j=0,1\,,\;d=1\,;\\
&f_j=\lim_{|x|\to0}\bigg[\psi_j(x)+\frac{q_j}{2\pi}\ln(|x|)\bigg]&&j=0,1\,,\;d=2\,;\\
&f_j=\lim_{|x|\to0}\bigg[\psi_j(x)-\frac{q_j}{4\pi|x|}\bigg]&&j=0,1\,,\;d=3\bigg\}\,.
\end{aligned}
\end{aligned}
\ee
%%%
%%%
\be
%\label{H0}
\hh_0
\begin{pmatrix}
 \psi_0\\
 \psi_1
 \end{pmatrix}
:=\begin{pmatrix}
-\Delta\phi_0^z+z\,q_0\, G^{z}\\
 (-\Delta+1)\phi_1^z+z\,q_1\, G^{z-1}
 \end{pmatrix}\,;\qquad\begin{pmatrix}
 \psi_0\\
 \psi_1
 \end{pmatrix} \in D(\hh_0)\,.
\ee
\end{definition}
For any choice of $\theta_0\in\RE$, $H_0$ belongs to a sub-family of the
self-adjoint extensions of the operator $S$ defined above. Each
Hamiltonian of the sub-family generates a dynamics where the spin evolution
is not affected by the particle ( \cite{cacciapuoti-carlone-figari:09} ). The two channel
Hamiltonian $H_0$ is often formally  written as
%%%
\be
%\label{H_0aghkh}
H_0=
\begin{pmatrix}
-\Delta+\al_0\delta&0\\ \\
0&-\Delta+1+\al_1\delta
\end{pmatrix}
\ee
%%%
to stress that it acts as $S$ on $D(S)$. To match our notation with the
one used in the monograph   \cite {aghh:05} on point interactions one should take into account the following
correspondence rules: $\al_j=-\theta_j$ for  $d=1$; $\al_j=1/\theta_j$ for
 $d=2,3$; $j=0,1$.

The spectrum of $H_0$ can be obtained directly from the spectrum of the
operator ``$-\Delta + \al_j\delta$'' ( see \cite{aghh:05}).  The main
results are collected in the following
%%%%%%%%%
%PROPOSITION
%%%%%%%%%
\begin{proposition}
\label{spectrumH0}
For $d=1,2,3$, the essential spectrum of $H_0$ fills the positive real
line, $\sigma_{ess}(H_0)=[0,\infty)$ and $0\in\sigma_{p}(H_0)$, the point
spectrum of $H_0$. The bound state  (of unit norm) corresponding to the zero
energy eigenvalue is $\Phi^{th}=(0,\phi_1^{th})$ with
%%%
\be
\phi_1^{th}(x)=
\left\{\begin{aligned}
&\, e^{-|x|}&\qquad&d=1\\
&\frac{\sqrt{\pi}}{2}\,  \Hbessel^{(1)}_0\big(i|x|\big)&&d=2\\
& \sqrt{2} \,\frac{e^{-|x|}}{4\pi|x|} && d=3
\end{aligned}\right.
\ee
%%%
Moreover:
\begin{itemize}
\item[-] For $d=1$ and  $\theta_0>0$, $\sigma_p(H_0)=\{-\theta_0^2/4,0\}$.
While for $\theta_0\leqslant 0$, $\sigma_p(H_0)=\{0\}$.
\item[-] For $d=2$ and $\theta_0\in\RE$, $\theta_0\neq0$,
$\sigma_p(H_0)=\{-e^{4\pi[1/a-1/\theta_0]},0\}$ (where the constant  $a$
was defined in equation \eqref{theta1}). While for $\theta_0= 0$,
$\sigma_p(H_0)=\{0\}$.
\item[-] For $d=3$ and   $\theta_0<0$,
$\sigma_p(H_0)=\{-(4\pi/\theta_0)^2,0\}$. While for $\theta_0\geqslant0$,
$\sigma_p(H_0)=\{0\}$.
 \end{itemize}
\end{proposition}

As noticed before $H_0$ describes a two independent channel system. In
each channel the particle ``feels" a point interaction placed in the origin
whose strength may depend on the channel (equivalently on the spin state).
Among all the self-adjoint extensions of $S$ one can find a large class of
Hamiltonians coupling the two channels. To the aim of examining the
behavior of the eigenvalue of $H_0$ at the threshold of the essential
spectrum, when the lower and upper channels are weakly coupled, we choose
a suitable Hamiltonian, $H_\ve$,  belonging to that class and close, in a
sense that will be made precise in the following, to $H_0$.
%%%%%%%%%
%DEFINITION
%%%%%%%%%
\begin{definition}
\label{def:Hve}
Let us take $\theta_0,b,c\in\RE$ and let $\ve>0$. For $d=1,2,3$ we set
%%%
\be
\theta_1^\ve=
\left\{\begin{aligned}
&2+c\ve&& d=1\\
&a+c\ve&& d=2\\
&-4\pi+c\ve&\quad& d=3
\end{aligned}
\right.
\qquad a=\frac{2\pi}{\ln(2)-\gamma}\,;
\ee
%%%
$\hh_{\ve}:D(\hh_{\ve})\subset \HH\to\HH$ is the self-adjoint operator:
\[
\begin{aligned}
D(\hh_\ve):=\bigg\{\Psi\equiv
\begin{pmatrix}
\psi_{0}\\
\psi_{1}
\end{pmatrix}\in\HH\bigg|&\psi_0=\phi_0^z+q_0
G^{z}\,,\;\psi_1=\phi_1^z+q_1 G^{z-1}\,;
\phi_0^z\,,\;\phi_1^z\in H^2(\RE^d)\,;\;z\in\CO\backslash\RE\,;\\
&q_0=\theta_0f_0+b\ve f_1\,,\;q_1=b\ve f_0+\theta_1^\ve f_1\,;\\
&
\begin{aligned}
&f_j=\psi_j(0)&&j=0,1\,,\;d=1\,;\\
&f_j=\lim_{|x|\to0}\bigg[\psi_j(x)+\frac{q_j}{2\pi}\ln(|x|)\bigg]&&j=0,1\,,\;d=2\,;\\
&f_j=\lim_{|x|\to0}\bigg[\psi_j(x)-\frac{q_j}{4\pi|x|}\bigg]&&j=0,1\,,\;d=3\bigg\}
\end{aligned}
\end{aligned}
\]
\[
\hh_\ve
\begin{pmatrix}
 \psi_0\\
 \psi_1
 \end{pmatrix}
:=\begin{pmatrix}
-\Delta\phi_0^z+z\,q_0\, G^{z}\\
 (-\Delta+1)\phi_1^z+z\,q_1\, G^{z-1}
 \end{pmatrix}\,;\qquad\begin{pmatrix}
 \psi_0\\
 \psi_1
 \end{pmatrix} \in D(\hh_\ve)\,.
\]
\end{definition}
For all $z\in\CO\backslash\RE$ we denote by $R_\ve(z):=(H_\ve-z)^{-1}$ the
resolvent of $\hh_\ve$. An explicit formula for $R_\ve(z)$ can be obtained
by using the theory of self-adjoint extensions of symmetric operators (see
\cite{cacciapuoti-carlone-figari:07} and \cite{cacciapuoti-carlone-figari:09}). We summarize the result in the following
formula
\begin{equation}
\label{Rve}
\begin{aligned}
\rr_\ve(z)=&
\begin{pmatrix}
(-\Delta-z)^{-1}&0\\ \\
0&(-\Delta+1-z)^{-1}
 \end{pmatrix}\\
&+\frac{1}{D_\ve(z)}
\begin{pmatrix}
\Gamma_{\ve,11}(z)(G^{\bar z},\cdot)G^{z}&&\Gamma_{\ve,12}(z)(G^{\bar
z-1},\cdot)G^{z}\\ \\
\Gamma_{\ve,21}(z)(G^{\bar z},\cdot)G^{z-1}&&\Gamma_{\ve,22}(z)(G^{\bar
z-1},\cdot)G^{z-1}
\end{pmatrix}
\end{aligned}
\end{equation}
where the function $G^{z}$ was defined in equation $\eqref{Gz}$ and
\begin{equation}
\label{den}
D_\ve(z)=
\left\{\begin{aligned}
&b^2\ve^{2}-(\theta_0+2 i \sqrt{z})\left[2(1+i\sqrt{z-1})+c\ve\right]& d=1\\
 &\big\{a+c\ve+\big[\theta_0(a+c\ve)-b^2\ve^2\big]g(z)\big\}\big\{\theta_0+\big[\theta_0(a+c\ve)-b^2\ve^2\big]g(z-1)\big\}-b^2\ve^2&
d=2\\
&\bigg[1-i\frac{\theta_0}{4\pi}\sqrt{z}\bigg]\bigg[1+i\left(1-\frac{c\ve}{4\pi}\right)\sqrt{z-1}\bigg]+\left(\frac{b\ve}{4\pi}\right)^2\sqrt{z-1}\sqrt{z}&d=3
\end{aligned}
\right.
\end{equation}
where $g(z):=\big[\ln(\sqrt{z})-i\pi/2\big]/(2\pi)-1/a$
and the matrix elements  $\Gamma_{\ve,ij}(z)$   read:\\
for d=1
%%%
\begin{equation}
\label{gamma1}
\begin{split}
&\Gamma_{\ve,11}(z)=-2
i\sqrt{z}\left[-b^2\ve^{2}+\theta_0\left(2+2i\sqrt{z-1}+c\ve\right)\right]\\
&\Gamma_{\ve,12}(z)=\Gamma_{\ve,21}(z)=4\,b\ve\,\sqrt{z-1}\,\sqrt{z}\\
&\Gamma_{\ve,22}(z)=-2i\sqrt{z-1}\left[-b^2\ve^{2}+(2+c\ve)(2i\sqrt{z}+\theta_0)\right],
\end{split}
\end{equation}
%%%

for d=2
%%%
\begin{equation}
\label{gamma2}
\begin{aligned}
&\Gamma_{\ve,11}(z)=\big\{\theta_0+\big[\theta_0(a+c\ve)-b^2\ve^2\big]g(z-1)\big\}
[\theta_0(a+c\ve)-b^2\ve^2]\\
&\Gamma_{\ve,12}(z)=\Gamma_{\ve,21}(z)=b\ve[\theta_0(a+c\ve)-b^2\ve^2]\\
&\Gamma_{\ve,22}(z)=\big\{a+c\ve+\big[\theta_0(a+c\ve)-b^2\ve^2\big]g(z)\big\}
[\theta_0(a+c\ve)-b^2\ve^2],
\end{aligned}
\end{equation}
%%%

for d=3
%%%
\begin{equation}
\label{gamma3}
\begin{split}
&\Gamma_{\ve,11}(z)=\theta_0\bigg[1+i\bigg(1-\frac{c\ve}{4\pi}\bigg)\sqrt{z-1}\bigg]-b^2\ve^2\frac{\sqrt{z-1}}{4\pi
i}\\
&\Gamma_{\ve,12}(z)=\Gamma_{\ve,21}(z)=b\ve\\
&\Gamma_{\ve,22}(z)=(-4\pi+c\ve)\bigg[1-i\frac{\theta_0}{4\pi}\sqrt{z}\bigg]-b^2\ve^2\frac{\sqrt{z}}{4\pi
i}\,.
\end{split}
\end{equation}
%%%

For $d=1,2,3$ the explicit form of the  resolvent of $H_0$,
$R_0(z):=(H_0-z)^{-1}$, can be obtained from formulas \eqref{Rve} -
\eqref{gamma3} by setting $\ve=0$.
%%%%%%%
%REMARK
%%%%%%%
\begin{remark}
\label{LEE}
For $\theta_0=0$ and  $d=1,2$ the resolvent $R_0(z)$ has a singularity
in $z=0$ on both channels. In the upper channel there is a polar
singularity corresponding to the eigenvalue. In the lower channel there is
a singularity of order $z^{-1/2}$ for $d=1$ and $\ln z$ for $d=2$
respectively. This is easily checked analyzing  the behavior around $z=0$
of $(-\Delta-z)^{-1} $ (see, e.g., \cite{jensen-nenciu:01}).    A precise
statement, obtained examining the integral kernel $G^z(x-y)$ of
$(-\Delta-z)^{-1} $, gives the following expansion
%%%
\be
(-\Delta-z)^{-1}=
\left\{
\begin{aligned}
&\frac{i}{2\sqrt{z}}+\OO(1)&\qquad&d=1
\\
&-\frac{1}{4\pi}\ln z+\OO(1)&\qquad&
d=2
\\
&\OO(1)&\qquad&d=3
\end{aligned}
\right.
\ee
%%%
where $\OO(1)$ denotes an operator on some suitable weighted $L^2$ space, whose norm remains bounded uniformly in $z$. A possible choice for the weighted space is for example $L^2(R^d, (1+|x|)^{-s}dx)$ for some $s$ large enough.
\end{remark}

Finally we notice that $H_{\ve}$ is a small perturbation of $H_{0}$ in the
resolvent sense, i.e., $\forall z\in\CO\backslash\RE$ there exists $\ve_0$
such that for all  $0<\ve<\ve_0$
\begin{equation*}
%\label{resext}
\|R_\ve(z)-R_0(z)\|_{\BB(\HH,\HH)}\leqslant \ve C
\end{equation*}
where  $C$ is a positive constant independent on $\ve$ and
$\|\cdot\|_{\BB(\HH,\HH)}$ is the operator norm in the vector space
${\BB(\HH,\HH)}$ of bounded linear operators on $\HH$.

%%%%%%%%%%%%%%%%%%%%%%%%%%%%%%%%%%%%%%%
%%%%%%%%%%%%%%%%%%%%%%%%%%%%%%%%%%%%%%%
%SECTION 
%%%%%%%%%%%%%%%%%%%%%%%%%%%%%%%%%%%%%%%
%%%%%%%%%%%%%%%%%%%%%%%%%%%%%%%%%%%%%%%
\section{
\label{sec2}
Results}

In this section we analyze the spectral structure  of $H_{\ve}$ to examine the effect of the coupling between the lower and the upper channels on the spectrum of the Hamiltonian $H_0$ and in particular on its zero energy eigenvalue. We denote by $\sigma_p(H_\ve)$, $\sigma_{ess}(H_\ve)$ and $\sigma_{ac}(H_\ve)$  the point, essential and absolutely continuous spectrum of $H_\ve$ respectively.

It will be clear from the proofs that value and sign of the parameter $b$ in definition \ref{def:Hve} do not affect our results in any substantial way. For this reason we set $b=1$.

We use sometimes the phrase  ``for $\ve$ small enough ... '' as a short version of ``there exists $\ve_0$ such that for all  $0<\ve<\ve_0$ ... ''.

%%%%%%%%%%%%%%%%%%%%%%%%%
%SUBSECTION
%%%%%%%%%%%%%%%%%%%%%%%%%
\subsection{Positive perturbations. $c<0$.}

In this section we study the behavior of the threshold eigenvalue when the parameter $c$ in the Hamiltonian $H_\ve$ is negative. This  choice corresponds to a positive perturbation of the Hamiltonian $H_0$ in the sense that for small $\ve$ the threshold eigenvalue is pushed inside the continuum as a consequence of the perturbative term  $c\ve$ in $\theta_1^\ve$. 

In order to make this statement more precise let us consider the case $b=0$ in definition \ref{def:Hve}. The Hamiltonian $H_\ve^{b=0}$ is a perturbation (in resolvent sense) of the Hamiltonian $H_0$ for which the channels $0$ and $1$ are not coupled. For all $\ve$ small enough the Hamiltonian $H_\ve^{b=0}$ has an eigenvalue $E_\ve^{b=0}$ in a ball of radius $\ve$ around the origin, and
\begin{align}
\label{b0d1}
&E_\ve^{b=0}=1-\frac{{\theta_1^\ve}^2}{4}=-c\ve-\frac{c^2\ve^2}{4}&&d=1\\
\label{b0d2}
&E_\ve^{b=0}=1-e^{4\pi(1/a-1/\theta_1^\ve)}=-\frac{4\pi c\ve}{a^2}+\OO(\ve^2)&&d=2\\
\label{b0d3}
&E_\ve^{b=0}=1-\frac{(4\pi)^2}{{\theta_1^\ve}^2}=-\frac{c\ve}{2\pi}+\OO(\ve^2)&&d=3\,.
\end{align}

Looking at $H_\ve$ as perturbation of $H_\ve^{b=0}$, we expect that the zero energy eigenvalue of $H_0$ will be driven in a resonance as it happens for embedded eigenvalues.

\noindent
We analyze first the cases $d=1,2$ and $\theta_0\neq 0$.
%%%%%%%
%THEOREM
%%%%%%%
\begin{theorem}
\label{teorema1}
Let $d=1,2$ and assume that $c<0$ and $\theta_0\neq0$. Then  there exists $\ve_0>0$ such that for all $0<\ve<\ve_0$:\\
the  essential  spectrum of $H_\ve$ fills the positive real line and is only absolutely continuous,
\begin{equation}
\label{th1ess}
\sigma_{ess}(H_\ve)=\sigma_{ac}(H_\ve)=[0,+\infty)\,;
\end{equation}
there exists a positive constant $C$ such that the Hamiltonian $H_\ve$ has no isolated eigenvalues in $(-C,0)$;\\
the analytic continuation of the resolvent $R_\ve(z)$ through the real axis from the semi-plane $\Im z >0$  has  a simple pole (resonance) in $z= E_{\ve}^{r}$  where $\Im (E_{\ve}^{r})<0$  and 
\begin{align}
\label{th1Eved1}
&E_\ve^{r}=
|c|\ve+\left(\frac{1}{\theta_0}-\frac{c^{2}}{4}\right)\ve^{2}-i\frac{2\sqrt{|c|}}{\theta_0^{2}}\ve^{\frac{5}{2}}+\OO(\ve^{4})&& d=1\\
\label{th1Eved2}
&
\Re\big(E_\ve^{r}\big)=\frac{4\pi|c|}{a^2}\ve+\OO(\ve^2)\;;\quad
\Im\big(E_\ve^{r}\big)=-\frac{16\pi^3}{|a\theta_0|^2}\frac{\ve^2}{|\ln\ve|^2}+o\big((\ve/|\ln\ve|)^2\big)&&d=2\,.
\end{align}
\end{theorem}
%%%%%%
%PROOF
%%%%%%
\begin{remark}
Theorem 1 does not characterize the entire spectral structure of the Hamiltonian $H_\ve$. Statements only concern spectral singularities in a small region of the complex plane around the origin. In particular for $\theta_0>0$ in $d=1$ and for $\theta_0\neq0$ in  $d=2$ there are negative eigenvalues close the corresponding eigenvalues of $H_\ve^{b=0}$. The relative singularities of the resolvent of $H_\ve$ are bounded away from the origin unless  $\theta_0$ is very small.

\noindent Looking at the dependence on $\theta_0 \rightarrow 0$ of the coordinates of the resonances in $d=1,2$ given in \eqref{th1Eved1} and \eqref{th1Eved2} one realizes that the case $\theta_0 =0$ has to be treated independently.
\end{remark}

\begin{proof}
We first consider the case $d=1$. We notice that  the singularity in $z=0$ of order $z^{-1/2}$, in the term $(-\Delta-z)^{-1}$ in formula \eqref{Rve} (see the remark \ref{LEE}),  is canceled by an opposite singularity arising from the coefficient $\Gamma_{\ve,11}/D_\ve$. In fact an explicit calculation gives  
%%%
\be
\frac{\Gamma_{\ve,11}(z)(G^{\bar z},\cdot)G^{z}}{D_\ve(z)}
=-\frac{i}{2\sqrt{z}}+\OO(1)
\ee
%%%
for all $\ve>0$, where the equality has  be intended in some weighted $L^2(\RE)$
 space, where the weight depends on the number of terms of the expansion one considers (see \cite{jensen-nenciu:01} for details). A similar remark holds true for the singularity in $z=1$ of order $(z-1)^{-1/2}$ in the upper channel. This singularity, arising from the term $(-\Delta+1-z)^{-1}$ in the formula \eqref{Rve}, is compensated by an opposite singularity in the term  $\Gamma_{\ve,22}/D_\ve$. Then the singularities of the resolvent on the real axes coincide with the zeros of the function $D_\ve(z)$ in equation \eqref{Rve}.  
 
We prove first statement \eqref{th1ess}, showing that $H_\ve$ has  no embedded eigenvalues or eigenvalues at the threshold. Let us set $z=\la>0$. Trivially $D_\ve(1)\neq0$ and $D_\ve(0)\neq0$. For $\la\in[0,1)$ a direct calculation shows that equations $\Im\big[D_\ve(\la)\big]=0$ and $\Re\big[D_\ve(\la)\big]=0$  are not compatible. We deduce that there are no solutions to the equation $D_\ve(\la)=0$ for $\la\in[0,1]$. For $\la>1$, taking real and imaginary part of the equation  $D(\la)=0$   we get
%%% 
 \begin{equation}
 \label{imd1}
 (2+c\ve)\sqrt{\la}=-\theta_0\sqrt{\la-1}
 \end{equation}
 %%%
and
%%%
\begin{equation}
\label{red1}
\ve^2+4\sqrt{\la}\sqrt{\la-1}=\theta_0(2+c\ve)\,.
\end{equation}
%%%
For $\theta_0>0$ equation \eqref{imd1} has no solutions in $(1,+\infty)$ and for $\theta_0<0$ equation \eqref{red1} has no solutions in $(1,+\infty)$.

Next we prove that there are no isolated eigenvalues in some suitable neighborhood of $z=0$. Let us set  $z=-\la$, then equation $D_\ve(-\la)=0$ gives 
\begin{equation}
\label{nonso}
-\theta_0 +2\sqrt{\la}=-\frac{\ve^{2}}{2(1-{\sqrt{\la+1}})-|c|\ve}\,.
\end{equation}
For $\la\in(0,+\infty)$ the right side of the equation  is a strictly positive and decreasing function which equals $\ve/|c|$ for $\la=0$. While the left hand side of the equation is a strictly increasing function which equals $-\theta_0$ for $\la=0$. Then for $\ve$ small enough there are no solutions of \eqref{nonso} when $\theta_0<0$. For $\theta_0>0$ the l.h.s. has a zero in $\la=\theta_0^2/4$ and it is negative in $(0,\theta_0^2/4)$ and positive in $(\theta_0^2/4,\infty)$. Then   for $\theta_0>0$ there is only one solution to equation \eqref{nonso}, say $\la_{0,\ve}$, and  $\la_{0,\ve}>\theta_0^2/4$. It follows that for $\theta_0\neq 0$ there are no isolated eigenvalues in $(-\theta_0^2/4,0)$.

To find a solution to the equation $D_\ve(z)=0$  we make use of the following recursive procedure. We first notice that the equation $D_\ve(z)=0$ can be written as
%%%
\be
i\sqrt{z-1}=-1-\frac{c\ve}{2}+\frac{\ve^2}{2(\theta_0+2i\sqrt{z})}.
\ee
%%%
We look then for a fixed point of the recurrence relation
%%%
\begin{align}
\label{th1d1p0}
&z^{(0)}=0\\
\label{th1d1p1}
&z^{(k+1)}=1-\bigg[1-\frac{|c|\ve}{2}-\frac{\ve^2}{2(\theta_0+2i\sqrt{z^{(k)}})}\bigg]^2\qquad k=0,1,2,...\,.
\end{align}
%%%
To prove convergence of the sequence $z^{(k)}$ we proceed by induction. Assume that $|z^{(k-1)}|\leqslant C\ve$, then equation \eqref{th1d1p1} implies that for $\ve$ small enough  $|\zz{k}|<C\ve$. Since $\zz{1}=|c|\ve+\OO(\ve^2)$ then  $|\zz{k}|\leqslant C\ve$ for all $k$. Moreover let us set $\zz{k}\equiv |c|\ve(1+\ww{k})$, $k=1,2,3,... $. Trivially  $|\ww{1}|\leqslant C\ve$ and by  equation \eqref{th1d1p1} 
\begin{equation}
\label{parlare}
\begin{aligned}
|\ww{k+1}-\ww k|=&\frac{1}{|c|\ve}\bigg|2\bigg(1-\frac{|c|\ve}{2}\bigg)\bigg(\frac{\ve^2}{2\big(\theta_0+2i\sqrt{|c|\ve}\sqrt{1+\ww k}\big)}-\frac{\ve^2}{2\big(\theta_0+2i\sqrt{|c|\ve}\sqrt{1+\ww{k-1}}\big)}\bigg)\bigg|\\
&+\frac{1}{|c|\ve}\bigg|\bigg(\frac{\ve^4}{4\big(\theta_0+2i\sqrt{|c|\ve}\sqrt{1+\ww k}\big)^2}-\frac{\ve^4}{4\big(\theta_0+2i\sqrt{|c|\ve}\sqrt{1+\ww{k-1}}\big)^2}\bigg)\bigg|\,.
\end{aligned}
\end{equation}
Since the function $(\theta_0+2i\sqrt{|c|\ve}{\sqrt{1+w}})^{-1}$ is analytic for $w$ in a ball of radius $\ve$ around the origin
\be
\bigg|\frac{1}{2\big(\theta_0+2i\sqrt{|c|\ve}\sqrt{1+\ww k}\big)}-\frac{1}{2\big(\theta_0+2i\sqrt{|c|\ve}\sqrt{1+\ww{k-1}}\big)}\bigg|\leqslant C\ve^{1/2}|\ww k-\ww{k+1}|\,.
\ee
The second term in the r.h.s. of equation \eqref{parlare} can be treated in a similar way. Then 
 $|\ww{k+1}-\ww{k}|\leqslant C\ve^{3/2}|\ww{k}-\ww{k-1}|$ for all $k=2,3,... $ which in turns implies that $|\zz{k+1}-\zz{k}|\leqslant C\ve^{3/2}|\zz{k}-\zz{k-1}|$ for all $k=2,3,... $; the sequence $\{\zz{k}\}$ converges in a ball of radius $\ve$ and 
\be
z^{(2)}=|c|\ve+\left(\frac{1}{\theta_0}-\frac{c^{2}}{4}\right)\ve^{2}-i\frac{2\sqrt{|c|}}{\theta_0^{2}}\ve^{\frac{5}{2}}+\OO(\ve^{4})\,.
\ee
\\

Let us now consider the case $d=2$. We notice that, similarly to what happens in the case $d=1$,  the logarithmic singularities in $z=0$ and $z=1$ due to $(-\Delta-z)^{-1}$ and $(-\Delta+1-z)^{-1}$  in equation \eqref{Rve} (see the remark \ref{LEE}) are compensated by opposite singularities arising from the coefficients $\Gamma_{\ve,11}/D_\ve$ and $\Gamma_{\ve,22}/D_\ve$. Then the singularities of the resolvent on the real axes  coincide with the zeros of the function $D_\ve(z)$.
 
First we prove the statement \eqref{th1ess}. Let us set $z=\la>0$ and analyze the equation  $D_\ve(\la)=0$ which can be written as 
 %%%
\be
\bigg[
a+c\ve+\big(\theta_0(a+c\ve)-\ve^2\big)\bigg(\frac{\ln\la}{4\pi}-\frac{1}{a}-\frac{i}4\bigg)
\bigg]
=\frac{\ve^2}{\theta_0+\big(\theta_0(a+c\ve)-\ve^2\big)\big(\frac{\ln(\la-1)}{4\pi}-\frac{1}{a}-\frac{i}4\big)
}\,.
\ee
%%% 
If  $\la>1$, taking the imaginary part  of both sides of the last equality  one gets after few manipulations
\be
-\big\{\theta_0+\big[\theta_0(a+c\ve)-\ve^2\big]\big(\ln(\la-1)/(4\pi)-1/a\big)\big\}^2=\frac{\big[\theta_0(a+c\ve)-\ve^2\big]^2}{16}+\ve^2\qquad\la>1
\ee
which obviously has no solutions. For $0<\la<1$ one can see that equations $
\Im(D_\ve(\la))=0$ and $\Re(D_\ve(\la))=0$ are not compatible. Moreover, as we already noticed, the resolvent $R_\ve(z)$ does not have singularities in $z=0$ and  $z=1$.  Then the equation $D_\ve(\la)=0$ has no solutions in $(0,\infty)$.

As a second step we look for isolated eigenvalues of $H_\ve$ in some suitable interval $(-C,0)$. For $\la>0$  $D_\ve(-\la)$ is real and equation $D_\ve(-\la)=0$ can be rearranged as 
\begin{equation}
\label{gold}
\big\{\theta_0+\big[\theta_0(a+c\ve)-\ve^2\big]g(-\la-1)\big\}=\frac{\ve^2}{a+c\ve+\big[\theta_0(a+c\ve)-\ve^2\big]g(-\la)}\qquad\la>0\,.
\end{equation}
Let us consider first the case $\theta_0>0$. The l.h.s. of equation \eqref{gold} is continuous and strictly increasing, moreover it  equals  $\theta_0 |c| \ve/a+\ve^2/a$ for $\la=0$. The function on the r.h.s. of the equation has a vertical asymptote in $\la_{a,\ve} =\exp\big[4\pi\big((\theta_0/a-1)(a-|c|\ve)-\ve^2/a\big)/\big(\theta_0(a-|c|\ve)-\ve^2\big)\big]$. It is continuous and strictly decreasing in $(0,\la_{a,\ve})$ and $(\la_{a,\ve},\infty)$. Finally, in $(\la_{a,\ve},\infty)$, it   is strictly positive and its limits in $\la=0$ and $\la=\infty$ equal zero. The above discussion makes clear that equation \eqref{gold} has always one  solution in $(\la_{a,\ve},\infty)$ and that such solution is contained in a ball of radius $\ve$ around $\la_{a,0}=\exp\big[4\pi(1/a-1/\theta_0)\big]$ which is the eigenvalue of the Hamiltonian $H_0$. For $c<0$ there are no solutions in $(0,\la_{a,\ve})$.  A similar result can be proved by an analogous argument if one assumes $\theta_0<0$.  We have then shown that for $c<0$ and $\theta_0\neq0$,   $H_\ve$ has no  isolated eigenvalues in $(-C,0)$ with $0<C<\la_{a,\ve}$.
 
Let us now prove that for $c<0$ and $\theta_0\neq0$ the equation $D_\ve(z)=0$ has a solution with negative imaginary part  in a neighborhood of $z=0$.  Again we make use  of a recursive procedure. The equation  $D_\ve(z)=0$  can be rearranged as follows
\begin{equation}
\label{Dzd2}
g(z-1)=\frac{\ve^2}{\big[\theta_0(a-|c|\ve)-\ve^2\big]\big\{a-|c|\ve+\big[\theta_0(a-|c|\ve)-\ve^2\big]g(z)\big\}}-\frac{\theta_0}{\big[\theta_0(a-|c|\ve)-\ve^2\big]}\,.
\end{equation}
Then for all $k=0,1,2,\dots$ we define the sequence  $\{z^{(k)}\}$ by
\begin{equation}
\label{recursive}
\begin{aligned}
&z^{(0)}=0\\
&\ln(1-z^{(k+1)})=r(\zz{k})\qquad k=0,1,2,...
\end{aligned}
\end{equation}
where $r(z)$ is the function
\be
%\label{rz}
r(z):=-\frac{4\pi\big[\theta_0|c|\ve+\ve^2\big]}{a\big[\theta_0(a-|c|\ve)-\ve^2\big]}+
\frac{4\pi\ve^2}{\big[\theta_0(a-|c|\ve)-\ve^2\big]\big\{a-|c|\ve+\big[\theta_0(a-|c|\ve)-\ve^2\big]g(z)\big\}}\,.
\ee
Taking the imaginary part of the equation \eqref{recursive} one has 
\begin{equation}
\arg(1-z^{(k+1)})=\frac{\ve^2}{\big|a-|c|\ve+\big[\theta_0(a-|c|\ve)-\ve^2\big]g(z^{(k)})\big|^2}
\big(\pi-\arg(z^{(k)})\big)\,.
\label{recursivearg}
\end{equation}
We remark  that  the Riemann surface associated to the  function $\arg(z)$ is made up of  an infinite number of sheets, each one labeled by an integer number $n$ and characterized by $\arg(z)\in[2n\pi,2(n+1)\pi)$. We call ``first'' (or ``physical'') Riemann sheet the one corresponding to $n=0$, i.e., $\arg(z)\in[0,2\pi)$. If $\arg(z)\in[-2\pi,0)$ we  say that $z$ is in the ``unphysical'' Riemann sheet corresponding to  $n=-1$.  In each sheet of the Riemann surface $\arg(\sqrt{z})=\arg(z)/2$. If $z$ is in the ``unphysical'' Riemann sheet corresponding to $n=-1$, then  $\Im(\sqrt{z})\leqslant0$.

In equation \eqref{recursivearg} the function $\arg(z^{(k)})$ has to be intended as the analytic continuation of the argument function through the cut $[0,\infty)$. Notice that we assumed $\arg(z^{(0)})=0$.

We proceed  by induction.  First we prove that there exists $\ve_0$ such that for all $0<\ve<\ve_0$, if $-\pi\leqslant\arg(z^{(k)})\leqslant0$ and $C_1\ve\leqslant|z^{(k)}|\leqslant C_2\ve$ then $-\pi\leqslant\arg(z^{(k+1)})\leqslant0$ and $C_3\ve\leqslant|z^{(k+1)}|\leqslant C_4\ve$, where the constants $C_3$ and $C_4$ do not  depend on $C_1$, $C_2$ and $\ve_0=\ve_0(C_1,C_2, C_3,C_4)$. We first notice that if  $-\pi\leqslant\arg(z^{(k)})\leqslant0$ and $C_1\ve\leqslant|z^{(k)}|\leqslant C_2\ve$, then  there exists $\ve_0$ such that for all $0<\ve<\ve_0$
\begin{equation}
\label{rainy}
\frac{4\pi\ve^2}
{\big|a-|c|\ve+\big[\theta_0(a-|c|\ve)-\ve^2\big]g(z^{(k)})\big|^2}\leqslant
C\frac{\ve^2}{|\ln\ve|^2}\,.
\end{equation}
From equation \eqref{recursivearg} and estimate \eqref{rainy} one immediately gets that for $\ve$ small enough  $0\leqslant\arg(1-z^{(k+1)})\leqslant C \ve^2/|\ln\ve|^2$, consequently $-\pi\leqslant\arg(z^{(k+1)})\leqslant0$. From the definition of the function $r(z)$ and from the inequality \eqref{rainy} it follows that $|r(\zz{k+1})|\leqslant C\ve$. Then from the analyticity of the exponential function one has
\begin{equation}
\label{zkk}
\zz{k+1}=1-e^{r(\zz{k})}=-r(\zz{k})+\OO(|r(\zz{k})|^2)
\end{equation}
which in turns implies
\begin{equation}
\label{43}
\frac{4\pi|c|}{a^2}\ve-C\ve^2\leqslant|\zz{k+1}|\leqslant\frac{4\pi|c|}{a^2}\ve+C\ve^2
\end{equation}
and one can take, e.g., $C_3\equiv 2\pi|c|\ve/a^2$ and $C_4\equiv  8\pi|c|\ve/a^2$.

We prove now that for all $k=1,2,...$,  $|\zz{k+1}-\zz k|\leqslant C \ve|\ln \ve|^{-2}|\zz k-\zz{k-1}|$. For $k =2,3,... $ let us set $\zz{k}\equiv (4\pi|c|\ve/a^2)(1+\ww{k})$.  By equation \eqref{zkk} it follows that, for any $k=1,2,3,...$, 
\be
|\zz{k+1}-\zz k|=
|e^{r(\zz{k})}-e^{r(\zz{k-1})}|\leqslant
C|r(\zz k)-r(\zz{k-1})|
\ee
where we used the analyticity of the exponential function. Equation \eqref{43} implies that $|\ww k|\leqslant C\ve $ for all $k=1,2,...$. The function $r((4\pi|c|\ve/a^2)(1+w))$ is analytic for $w$ in a ball of radius $\ve$ around the origin, and
\be
\begin{aligned}
&|r(\zz k)-r(\zz{k-1})|=
|r((4\pi|c|\ve/a^2)(1+\ww k))-r((4\pi|c|\ve/a^2)(1+\ww{k+1}))|\\
\leqslant&
C \frac{\ve^2}{|\ln \ve|^2}|\ww k-\ww{k-1}|=
C \frac{\ve}{|\ln \ve|^2}|\zz k-\zz{k-1}|\,.
\end{aligned}
\ee
Since  $\arg(\zz1)=0$ and $\zz1=4\pi|c|\ve/a^2+\OO(\ve^2)$ we have proved that   the sequence $\zz{k}$ is convergent in a ball of radius $\ve$. The lowest order of $E_\ve^{r}$ given in equation \eqref{th1Eved2}  can be computed  by using formula \eqref{recursive}.
\end{proof}

In the next theorem we analyze the behavior of the zero energy eigenvalue when $\theta_0=0$, $c<0$ for $d=1,2$. In addition to the eigenvalues in \eqref{b0d1} and  \eqref{b0d2}, $H_\ve^{b=0}$ has in this case a resonance at zero energy in the lower channel. The coupling between channels will turn the resonance into a negative eigenvalue.

%%%%%%%
%THEOREM
%%%%%%%
\begin{theorem}
\label{teorema2}
Let $d=1,2$, assume that $c<0$ and set $\theta_0=0$. Then there exists $\ve_0>0$ such that for all $0<\ve<\ve_0$:\\
the  essential  spectrum of $H_\ve$  fills the positive real line and  is only absolutely continuous,
\begin{equation}
\label{th2ess}
\sigma_{ess}(H_\ve)=\sigma_{ac}(H_\ve)=[0,+\infty)\,;
\end{equation}
the analytic continuation of the resolvent $R_\ve(z)$ through the real axis from the semi-plane $\Im z >0$  has  a simple pole (resonance) in $z= E_{\ve}^{r}$ where $\Im (E_{\ve}^{r})<0$ and
\begin{align}
\label{th2Everd1}
&E_{\ve}^{r}=|c|\ve-\frac{i}{2\sqrt{|c|}}\ve^{\frac{3}{2}}+\OO(\ve^{2})
&&d=1\\
\label{th2Everd2}
&\Re(E_\ve^{r})=\frac{4\pi|c|}{a^2}\ve+\OO\big(\ve^2|\ln\ve|\big)\;;\quad
\Im(E_\ve^{r})=-\frac{\pi}{a^{2}}\ve^{2}+o(\ve^2)&\quad&d=2\,;
\end{align}
the Hamiltonian $H_\ve$ has an isolated eigenvalue $E_\ve<0$ and
\begin{align}
\label{th2ied1}
&-\frac{\ve^2}{4|c|^2}<E_\ve<0&&d=1\\
\label{th2ied2}
&-\exp\bigg(-\frac{4\pi|c|}{\ve}+\frac{4\pi}{a}\bigg)<E_{\ve}<0&&d=2\,.
\end{align}
\end{theorem}
%%%%%%%
%REMARK
%%%%%%%
\begin{remark}
\label{menoinfinito}
For $d=2$ and $\theta_0=0$ the Hamiltonian $H_\ve$ has also one eigenvalue which moves to minus infinity as $\ve$ goes to zero. One can interpret the latter eigenvalue as the effect of the perturbation on the ``eigenvalue" of $H_0$ at $-\infty$ in $\theta_0=0^{-}$ (see Proposition 1). The emergence of the two eigenvalues \eqref{th2ied1}, \eqref{th2ied2}, respectively in $d=1$ and $d=2$, reveals the effect of the channel coupling on the resonances at zero energy in the lower channels, whereas the two resonances \eqref{th2Everd1}, \eqref{th2Everd2} appear as the effect of perturbing the eigenvalues in the upper channels.
\end{remark}
%%%%%%
%PROOF
%%%%%%
\begin{proof}
We notice first that also in this case, similarly to the case $\theta_0\neq0$ (see theorem \ref{teorema1}), the singularities in $z=0$ and $z=1$, arising from $(-\Delta-z)^{-1}$ and $(-\Delta+1-z)^{-1}$ in equation \eqref{Rve}, are canceled by opposite singularities in the two terms $\Gamma_{\ve,11}/D_\ve$ and $\Gamma_{\ve,22}/D_\ve$ respectively. This is true both in $d=1$ and for $d=2$.

Statement \eqref{th2ess} can be proved as it was done in theorem \ref{teorema1}. A straightforward analysis shows in fact that equation $D_\ve(\la)=0$ has no solutions for $\la\in[0,\infty)$.

We discuss now the existence of isolated eigenvalues in $d=1$. Let us set $z=-\la$, $\la>0$. The equation $D_\ve(-\la)=0$ can be written as 
%%%
\begin{equation}
\label{eqes33}
2\sqrt{\la}=\frac{\ve^2}{2(\sqrt{\la+1}-1)+|c|\ve}\,.
\end{equation}
The left hand side in the last equation is a strictly positive, increasing function which equals zero for $\la=0$. The right hand side of the equation is a strictly decreasing, positive function which equals $\ve/|c|$ for $\la=0$. It is obvious that there is only one solution $\la_\ve$ to the equation \eqref{eqes33} and that $0<\la_\ve<\ve^2/(4|c|^2)$ which in turn implies estimate \eqref{th2ied1}.
 
We prove now the existence of an isolated eigenvalue of $H_\ve$ for $d=2$. For $\la>0$ the  equation $D_\ve(-\la)=0$ can be written as 
\begin{equation}
\label{louisiana}
\frac{\ln(1+\la)}{4\pi}=\frac{1}{a}-\frac{1}{a+c\ve-\ve^2\big(\ln \la/(4\pi)-1/a\big)}\qquad \la>0\,.
\end{equation}
Let us denote by $f^r(\la)$ the r.h.s. of the last  equation,   
\be
%\label{black}
f^r(\la):=\frac{c\ve+\ve^2/a-\ve^2\ln \la/(4\pi)}{a\big[a+c\ve+\ve^2/a-\ve^2\ln \la/(4\pi)\big]}\,.
\ee
The l.h.s. of  equation \eqref{louisiana} is a strictly positive  and increasing function which equals zero for $\la=0$. The function $f^r$ has a vertical asymptote in $\la_{a,\ve}=\exp\big[4\pi\big(a/\ve^2-|c|/\ve+1/a\big)\big]$ and is strictly decreasing in $(0,\la_{a,\ve})$ and $(\la_{a,\ve},\infty)$; we notice that $\la_{a,\ve}$ goes to infinity as $\ve$ goes to zero. Moreover the function $f^r$ equals $1/a$ for $\la=0$ and zero for $\la=\la_{\ve}=\exp\big[4\pi\big(-|c|/\ve+1/a\big)\big]$. It is positive for $\la\in(0,\la_{\ve})$ and negative for $\la\in(\la_{\ve},\la_{a,\ve})$.
 
 Then there are two solutions to equation \eqref{louisiana}. One is in the interval $(\la_{a,\ve},\infty)$. The  isolated eigenvalue corresponding to this solution moves towards minus infinity as $\ve$ goes to zero. The other solution is in $(0,\la_{\ve})$. The corresponding eigenvalue satisfies estimate \eqref{th2ied2}.\\

We investigate now the presence of poles of the resolvent in the ``unphysical'' Riemann sheet.  Let us consider first the case $d=1$, similarly to what was done in the previous theorem we use a recursive procedure. We rewrite the equation $D_\ve(z)=0$ as 
%%%
\begin{equation}
\label{feed}
i\sqrt{z-1}=-1+
\frac{|c|\ve}{2}
+\frac{\ve^2}{4i\sqrt{z}}\,.
\end{equation}
%%%
Which implies 
\be
 z=1-\bigg[1-
\frac{|c|\ve}{2}
-\frac{\ve^2}{4i\sqrt{z}}\bigg]^2\,.
\ee
We set $z=|c|\ve(1+w)$ and  define the recursive procedure 
%%%
\begin{align*}
&\ww{0}=0\\
&\ww{k+1}=\frac{\ve^{1/2}}{2|c|i\sqrt{|c|(1+\ww{k})}}-\frac{|c|\ve}{4}
-\frac{\ve^{3/2}}{4i\sqrt{|c|(1+\ww k)}}+\frac{\ve^2}{16|c|^2(1+\ww k)}
\qquad k=0,1,2,...\,,
\end{align*}
%%%
then $\zz k=|c|\ve(1+\ww k)$ for all $k=0,1,2, ...$. By induction it is easy to prove that for all $k=0,1,2,...$ and for $\ve$ small enough, $|\ww{k}|\leqslant C\ve^{1/2}$ which in turns implies $|\zz k|\leqslant C\ve$. Moreover for all $k$,  $|\ww{k+1}-\ww{k}|\leqslant C \ve^{1/2}|\ww{k}-\ww{k-1}|$. Then $|\zz{k+1}-\zz{k}|\leqslant C \ve^{1/2}|\zz{k}-\zz{k-1}|$, consequently   the sequence $\{\zz{k}\}$  is convergent in a ball of radius $C\ve$ and the solution of the equation \eqref{feed} can be written as $E_\ve^r=\zz\infty$. By a straightforward computation it is easy to see that 
%%%
\be
E_{\ve}^{r}=|c|\ve-\frac{i}{2\sqrt{|c|}}\ve^{\frac{3}{2}}+\frac{|c|^{2}}{4}\ve^{2}+\OO(\ve^{3})\,.
\ee
%%%

In dimension $d=2$ the existence of a resonance can be proved in a similar way. The  equation $D_{\ve}(z)=0$  can be rearranged as 
\be
g(z-1)=-\frac{1}{a-|c|\ve-\ve^2g(z)}
\ee
which implies 
\begin{equation}
\label{mannish}
\ln(1-z)=4\pi\frac{|c|\ve+\ve^2g(z)}{a[|c|\ve+\ve^2g(z)-a]}=:r(z)\,.
\end{equation}
We prove that, if $c<0$, there exists a solution (with negative imaginary part) of the last equation  in a neighborhood of radius of order $\ve^2|\ln(\ve)|$ of the point $4\pi|c|\ve/a^2$.  Let us pose $z=4\pi|c|\ve(1+w)/a^2$ and let us define the sequence $\{w^{(k)}\}$ by
\begin{align*}
&w^{(0)}=0\\
&\ww{k+1}=-1+\frac{a^2}{4\pi|c|\ve}\big(1-e^{\widetilde r(\ww k)}\big)
 \qquad k=0,1,2,...
\end{align*}
where 
\begin{equation}
\label{tilder}
\widetilde r(w):=r\Big(\frac{4\pi|c|\ve}{a^2}(1+w)\Big)\,.
\end{equation}
Taking into account the analyticity of the exponential function one obtains 
\[
w^{(k+1)}
=-1-\frac{a^2}{4\pi|c|\ve}\widetilde r(w^{(k)})+\OO(|\widetilde r(\ww k)|^2/\ve)
\qquad
k=0,1,2,...\;.
\]
Noticing that for $|\ww k|\leqslant C \ve |\ln \ve|$ the following expansion holds true
\be
\widetilde r(\ww k)=-\frac{4\pi|c|\ve}{a^2}+\OO(\ve^2|\ln \ve|)\,.
\ee
By induction one can prove that $|\ww k|\leqslant C \ve |\ln\ve|$ for all $k=0,1,2,...$ .

Moreover by using again the analyticity of the exponential function and the fact that $\widetilde r(w)$ is analytic in a ball of radius $\ve|\ln \ve|$ one can see that 
\[
|w^{(k+1)}-w^{(k)}|\leqslant \frac{C}{\ve}|\widetilde r(\ww k)-\widetilde r(\ww{k-1})|
\leqslant C\ve|w^{(k)}-w^{(k-1)}|\,.
\]
Then the sequence $\{\ww k\}$ converges in a ball of radius $\ve |\ln \ve|$ and $\{z^{(k)}\}$ converges in a ball of radius $\ve^2|\ln \ve|$ around the point $4\pi|c|\ve/a^2$. Since $\Im (\widetilde r(0))=\pi\ve^2/a^2+o(\ve^2)$ (see equations \eqref{mannish} and   \eqref{tilder}) to the lowest order we get 
\[
\Im(\ww 1)=-\frac{a^2}{4\pi|c|\ve}\sin\big(\Im (\widetilde r(0))\big)+\OO(|\widetilde r(0)|^2/\ve)=
-\frac{a^2}{4\pi|c|\ve}\pi\ve^2/a^2+o(\ve)\,,
\]
which implies the second estimate estimate in \eqref{th2Everd2}. 
\end{proof}

In the next theorem we analyze the behavior of the zero energy eigenvalue when $c<0$ for $d=3$. 
%%%%%%%
%THEOREM
%%%%%%%
\begin{theorem}
%\label{teorema3}
Let $d=3$ and assume that $c<0$ then  there exists $\ve_0>0$ such that for all $0<\ve<\ve_0$:\\
the  essential  spectrum of $H_\ve$  fills the positive real
line and is only absolutely continuous,
\begin{equation}
\label{essd3}
\sigma_{ess}
(H_\ve)=\sigma_{ac}(H_\ve)=[0,+\infty)\,;
\end{equation}
there exists a positive constant $C$ such that the Hamiltonian $H_\ve$ has no isolated eigenvalues in $(-C,0)$;\\
the analytic continuation of the resolvent $R_\ve(z)$ through the real axis from the semi-plane $\Im z >0$  has  a simple pole (resonance) in $z= E_{\ve}^{r}$ where $\Im (E_{\ve}^{r})<0$ and 
\begin{equation}
\label{th3Ever}
E_\ve^r=\frac{|c|\ve}{2\pi}-\frac{3|c|^2\ve^2}{16\pi^2}-i\frac{1}{8\pi^2}\sqrt{\frac{|c|}{2\pi}}\ve^{5/2}+\OO(\ve^3)\,.
\end{equation}
\end{theorem}
%%%%%%
%PROOF
%%%%%%
\begin{proof}
Using the classical results we summarized in remark \ref{LEE} and the explicit expressions for the components of the matrix $\Gamma_\ve$ (see equation \eqref{gamma3}) one can conclude that  for $d=3$ the singularities of the resolvent $R_\ve(z)$  coincide with the zeros of the function $D_\ve(z)$ defined in equation \eqref{den}.

By direct analysis one can see that the equation $D_\ve(\la)=0$ has no solutions for $\la>0$ from which the  statement \eqref{essd3} directly follows.

The eigenvalues of the Hamiltonian $H_\ve$ are given by the solutions of the equation $D_\ve(-\la)=0$ for $\la>0$. Such equation can be rearranged as 
\begin{equation}
\label{all}
1+\frac{\theta_0}{4\pi}\sqrt{\la}=\bigg[\Big(1+\frac{\theta_0}{4\pi}\sqrt{\la}\Big)\Big(1-\frac{c\ve}{4\pi}\Big)+\frac{\ve^2}{(4\pi)^2}\sqrt{\la}\bigg]\sqrt{1+\la}\qquad \la>0\,.
\end{equation}
For $c<0$ and $\theta_0\geqslant0$ equation \eqref{all} has no solutions while for  $c<0$ and $\theta_0<0$ there is only one solution in a neighborhood of radius $\ve$ of the point $(4\pi/\theta_0)^2$. Then for $c<0$ and $\ve$ small enough there are no isolated eigenvalues in some suitable neighborhood of the origin. 

It remains to analyze the existence of solutions of the equation $D_\ve(z)=0$ in  the unphysical Riemann sheet in a neighborhood of the origin. The equation $D_\ve(z)=0$ can be rearranged in the following way
\begin{equation}
\label{63}
\bigg[1+i\left(1-\frac{c\ve}{4\pi}\right)\sqrt{z-1}\bigg] =-\frac{\ve^2}{(4\pi)^2}
\bigg[1-i\frac{\theta_0}{4\pi}\sqrt{z}\bigg]^{-1}\sqrt{z-1}\sqrt{z}\,.
\end{equation}
 From which we define the sequence $\{\zz k\}$: $\zz 0=0$,
\be
\begin{aligned}
\zz{k+1}=&
\frac{-c\ve/(2\pi)+c^2\ve^2/(4\pi)^2}{[1-c\ve/(4\pi)]^2}
-\frac{2\ve^2}{(4\pi)^2}\frac{\sqrt{\zz{k}-1}\,\sqrt{\zz k}}{\big(1-c\ve/(4\pi)\big)^2
\big(1-i\theta_0\sqrt{\zz k}/(4\pi)\big)}\\
&-\frac{\ve^4}{(4\pi)^4}\frac{(\zz{k}-1)\zz k}{\big(1-c\ve/(4\pi)\big)^2
\big(1-i\theta_0\sqrt{\zz k}/(4\pi)\big)^2}
\bigg]^2
\end{aligned}
\quad k=0,1,2,\dots,
\ee
where the formula for $\zz{k+1}$ was obtained by solving for $z$ the l.h.s. of equation \eqref{63}.
By induction one can prove that $|\zz k|\leqslant C\ve$ for all $k=0,1,2,...$. Moreover by direct computation one can see that $|\zz{k+1}-\zz k|\leqslant C\ve^{3/2}|\zz k-\zz{k-1}|$. Then the series $\{\zz k\}$ converges in a ball of radius $\ve$ and $E_\ve^r\equiv\zz\infty$; by direct computation one can prove expansion \eqref{th3Ever}. 
\end{proof}

%%%%%%%%%%%%%%%%%%%%%%%%%
%SUBSECTION
%%%%%%%%%%%%%%%%%%%%%%%%%
\subsection{Pure off-diagonal perturbations. $c=0$.}

The case $c=0$ marks the boundary between two different behaviors of  the threshold eigenvalue under perturbation: the evolution towards a proper eigenvalue $(c>0)$ and the evolution towards a resonance $(c<0)$. The model shows, in this case, peculiar features, strongly depending on the spatial dimension, which we will analyze for $d=1,2,3$ separately.

If in addition to $c=0$ we set $\theta_0=0$ the Hamiltonian $H_0$ shows, for $d=1,2$, both a zero energy eigenvalue in the upper channel and a zero energy resonance in the lower one. Once more the case  $\theta_0=0$ requires a distinct analysis.

As it was done in the previous sections we set $b=1$.

We study first  the case $d=1$.
%%%%%%%
%THEOREM
%%%%%%%
\begin{theorem}
%\label{teorema4}
Let $d=1$ and assume that $c=0$. Then  there exists $\ve_0>0$ such that for all $0<\ve<\ve_0$:\\
the  essential  spectrum of $H_\ve$  fills the positive real line and  is only absolutely continuous,
\begin{equation}
\label{th4ess}
\sigma_{ess}(H_\ve)=\sigma_{ac}(H_\ve)=[0,+\infty)\,.
\end{equation}
Moreover for $\ve$ small enough: 
\begin{itemize}
\item[-] if $\theta_0>0$ the analytic continuation of the resolvent $R_\ve(z)$ through the real axis from the semi-plane $\Im z >0$  has  a simple pole (resonance) in $z= E_{\ve}^{r}$ and 
\begin{equation}
\label{th4Ever}
E_\ve^r=\frac{\ve^2}{\theta_0}-i\frac{2}{\theta_0^{5/2}}\ve^3+\OO(\ve^4)\,,
\end{equation}
there exists a positive constant $C$ such that the Hamiltonian $H_\ve$ has no isolated eigenvalues in $(-C,0)$;\\
\item[-] if $\theta_0=0$ the Hamiltonian $H_\ve$ has an isolated eigenvalue in $z=E_\ve$ and 
\begin{equation}
\label{th4Eveth0}
E_\ve=-\frac{\ve^{4/3}}{2^{2/3}}+\OO(\ve^{8/3})\,;
\end{equation}
\item[-] if $\theta_0<0$ the Hamiltonian $H_\ve$ has  an isolated eigenvalue in $z=E_\ve$ and 
\begin{equation}
\label{th4Evethpos}
E_\ve=-\frac{\ve^2}{|\theta_0|}+\OO(\ve^3)\,.
\end{equation}
\end{itemize}
\end{theorem}
%%%%%%
%PROOF
%%%%%%
\begin{remark}
\label{photo}
When $\theta_0=0$ the Hamiltonian $H_\ve$ shows also two singularities at the same distance from the origin of the eigenvalue \eqref{th4Eveth0} and with $\arg(z = -5 \pi/3)$ and $\arg(z = -\pi/3)$. 
%To the latter is possible to associate a resonant state.
\end{remark}

\begin{proof}
Similarly to the case $c<0$, see theorems \ref{teorema1} and \ref{teorema2}, for any $ \theta_0\in\RE$ the singularities in $z=0$ and $z=1$, arising from $(-\Delta-z)^{-1}$ and $(-\Delta+1-z)^{-1}$,  are compensated by the terms $\Gamma_{\ve,11}/D_\ve$ and $\Gamma_{\ve,22}/D_\ve$ respectively.
 Then the statement \eqref{th4ess} is a consequence of the fact that the equation $D_\ve(\la)=0$ has no solutions for $\la>0$ and  the singularities of the resolvent coincide with the roots  of the equation $D_\ve(z)=0$.\\

For $\theta_0>0$ the equation $D_\ve(-\la)=0$, for $\la>0$,  has only one solution  in $(\theta_0^2/4,\infty)$, then the Hamiltonian $H_\ve$ has only one isolated eigenvalue in $(-\infty,-\theta_0^2/4)$. 

\noindent
The existence of a pole in the ``unphysical'' Riemann sheet can be proven by making use of a recursive procedure. The equation $D_\ve(z)=0$ can be rearranged as 
\be
i\sqrt{z-1}+1=\frac{\ve^2}{2(\theta_0+2i\sqrt{z})}\,.
\ee
To find the solution of the last equation we define the recursive procedure 
\begin{align*}
&\zz 0=0\\
&\zz{k+1}=\frac{\ve^2}{\theta_0+2i\sqrt{\zz k}}-\frac{\ve^4}{4(\theta_0+2i\sqrt{\zz k})^2}\qquad k=0,1,2,...\,.
\end{align*}
By using techniques similar to the ones used in the proof of theorem \ref{teorema1} one can see that the sequence $\{\zz k\}$ is convergent in a ball of radius $\ve^2$ around the origin and that the estimate \eqref{th4Ever} holds.\\

For $\theta_0=0$ the equation $D_\ve(z)=0$ reads
\begin{equation}
\label{rider}
4i\sqrt{z}(1+i\sqrt{z-1})-\ve^2=0\,.
\end{equation}
To find the isolated eigenvalue we set $z=-\la$, $\la>0$, and rearrange the equation \eqref{rider} as 
\be
4(\sqrt{1+\la}-1)=\frac{\ve^2}{\sqrt{\la}}\qquad \la>0\,.
\ee
Obviously the last equation has only one solution. The recursive procedure
\begin{align*}
&\lla 0=\frac{\ve^{4/3}}{2^{2/3}}\\
&\lla{k+1}=\frac{\ve^2}{2\sqrt{\lla k}}+\frac{\ve^4}{16\lla k}\qquad k=0,1,2,...
\end{align*}
converges to the solution  and can be used to  prove the estimate \eqref{th4Eveth0}.\\

For $\theta_0<0$ the existence of an isolated  isolated eigenvalue and the estimate \eqref{th4Evethpos} can be proven by writing the equation $D_\ve(-\la)=0$, $\la>0$, as
\be
2(\sqrt{1+\la}-1)=\frac{\ve^2}{2\sqrt{\la}+|\theta_0|}
\ee
and by using the recursive procedure defined by 
\begin{align*}
&\lla 0=0\\
&\lla{k+1}=\frac{\ve^2}{2\sqrt{\lla k}+|\theta_0|}+\frac{\ve^4}{4(2\sqrt{\lla k}+|\theta_0|)^2}\qquad k=0,1,2,...\,.
\end{align*}
\end{proof}

Let us now analyze the case $d=2$.
%%%%%%%
%THEOREM
%%%%%%%
\begin{theorem}
\label{teorema5}
Let $d=2$ and assume that $c=0$. Then  there exists $\ve_0>0$ such that for all $0<\ve<\ve_0$:\\ 
the  essential  spectrum of $H_\ve$   fills the positive real line and   is only absolutely continuous,
\begin{equation}
\label{th5ess}
\sigma_{ess}(H_\ve)=\sigma_{ac}(H_\ve)=[0,+\infty)\,.
\end{equation}
Moreover for $\ve$ small enough: 
\begin{itemize}
\item[-] if $\theta_0>0$ the analytic continuation of the resolvent $R_\ve(z)$ through the real axis from the semi-plane $\Im z >0$  has  a simple pole (resonance) in $z= E_{\ve}^{r}$ and 
\begin{equation}
\label{th5Ever}
\Re\big(E_\ve^r\big)=\frac{4\pi\ve^2}{a^2\theta_0}+\OO\big(\ve^2/|\ln\ve|\big)\;;\qquad
\Im\big(E_\ve^r\big)=-\frac{4\pi^3}{a^2\theta_0^2}\frac{\ve^2}{|\ln \ve|^2}+o\big(\ve^2/|\ln \ve|^2\big)\,,
\end{equation}
there exists a positive constant $C$ such that the Hamiltonian $H_\ve$ has no isolated eigenvalues in $(-C,0)$\,;\\
\item[-] if $\theta_0=0$ the Hamiltonian $H_\ve$ has an isolated eigenvalue in $z=E_\ve$ and 
\begin{equation}
\label{th5Eveth0}
E_\ve=-\frac{2}{a^2}\ve^2|\ln \ve|+o\big(\ve^2/|\ln \ve|\big)\,;
\end{equation}
\item[-] if  $\theta_0<0$ the Hamiltonian $H_\ve$ has  an isolated eigenvalue in $z=E_\ve$ and 
\begin{equation}
\label{th5Evethpos}
E_\ve=-\frac{4\pi\ve^2}{a^2|\theta_0|}+\OO\big(\ve^2/|\ln\ve|\big)\,.
\end{equation}
\end{itemize}
\end{theorem}
%%%%%%%
%REMARK
%%%%%%%
\begin{remark}
\label{bargain}
For $d=2$, $c=0$  and $\theta_0=0$ the Hamiltonian $H_\ve$ has also one eigenvalue which moves to minus infinity as $\ve$ goes to zero.
\end{remark}
%%%%%%
%PROOF
%%%%%%
\begin{proof}
Similarly to the case $c<0$, see theorems \ref{teorema1} and \ref{teorema2}, for any $ \theta_0\in\RE$ the singularities in $z=0$ and $z=1$, arising from $(-\Delta-z)^{-1}$ and $(-\Delta+1-z)^{-1}$,   are absorbed by the terms $\Gamma_{\ve,11}/D_\ve$ and $\Gamma_{\ve,22}/D_\ve$ respectively.  Then the statement \eqref{th5ess} is a consequence of the fact that the equation $D_\ve(\la)=0$ has no solutions for $\la>0$ and  the singularities of the resolvent coincide with the roots  of the equation $D_\ve(z)=0$.\\

For $\theta_0>0$ the equation $D_\ve(-\la)=0$, for $\la>0$,  has only one solution  in $(\la_{a,\ve},\infty)$, with $\la_{a,\ve} =\exp\big[4\pi\big(\theta_0-a-\ve^2/a\big)/\big(\theta_0a-\ve^2\big)\big]$; this can be seen by setting $c=0$ in equation \eqref{gold}. Then the Hamiltonian $H_\ve$ has only one isolated eigenvalue in $(-\infty,-\la_{a,\ve})$. 

To prove the existence of the pole $E_\ve^r$ in the unphysical Riemann sheet we proceed as we did in theorem \ref{teorema1}. We set $c=0$ in equation \eqref{Dzd2} and  define the sequence
\be
\begin{aligned}
&\zz 0=0\\
&\zz{k+1}=1-e^{r(\zz k)}\qquad k=0,1,2,...
\end{aligned}
\ee
with 
\be
r(z):=-\frac{4\pi\ve^2}{a}\,
\frac{\big(\ln z/(4\pi)-1/a-i/4\big)}{a+(a\theta_0-\ve^2)\big(\ln z/(4\pi)-1/a-i/4\big)}\,.
\ee
Following what was done in the proof of theorem \ref{teorema1} one can prove that the sequence $\{\zz k\}$ converges in a ball of radius $\ve^2$ and $E_\ve^r\equiv\zz\infty$ is a solution of the equation $D_\ve(z)=0$. The estimate \eqref{th5Ever} can be obtained by computing $\zz 2$.\\

To find the eigenvalues of $H_\ve$ for $c=0$ and $\theta_0=0$ we set $z=-\la$, $\la>0$, and study the equation $D_\ve(-\la)=0$ which reads 
\be
\bigg(\frac{\ln(1+\la)}{4\pi}-\frac{1}a\bigg)\bigg(a-\frac{\ve^2\ln \la}{4\pi}+\frac{\ve^2}a\bigg)=-1\qquad \la>0\,.
\ee
This equation can be rearranged as 
\begin{equation}
\label{santiago}
\frac{\ln(1+\la)}{4\pi}=
\frac{1}a+\frac{1}{\frac{\ve^2\ln \la}{4\pi}-a-\frac{\ve^2}a}\qquad\la>0\,.
\end{equation}
The analysis of the r.h.s. and of the l.h.s. of the equation as functions of $\la$ is trivial and one can refer to what was done in theorem \ref{teorema1} in the study of equation \eqref{gold}.  It comes out that the last equation has two solutions. One is in $(\exp(4\pi(a/\ve^2+1/a)),\infty)$. The eigenvalue of $H_\ve$ associated to this solution goes to minus infinity as $\ve$ goes to zero (see remark \ref{bargain}). The other solution is in $(0, \exp(4\pi/a))$ and we denote it by $\la_\ve$. To obtain an estimate of $\la_\ve$  we define the  sequence
\be
%\label{niente}
\begin{aligned}
&\lla 0=0\\
&\frac{\ln(1+\lla{k+1})}{4\pi}=
\frac{1}a+\frac{1}{\frac{\ve^2\ln \lla k}{4\pi}-a-\frac{\ve^2}a} \qquad k=0,1,2,...\,.
\end{aligned}
\ee
The convergence of the sequence $\{\lla k\}$ to the fixed point $\la_\ve$ can be deduced by examining the functions appearing on the two sides of equation \eqref{santiago}.  Such analysis will not be detailed here. We limit ourselves to compute  
the terms which are needed to get the estimate \eqref{th5Eveth0}. By direct computation 
\be
\lla 1=e^{4\pi/a}-1<e^{4\pi/a}\,,
\ee
\be
\lla 2=\exp\Big(\frac{4\pi\ve^2}{\widetilde a^2}+\OO(\ve^4)\Big)-1=
\frac{4\pi\ve^2}{\widetilde a^2}+\OO(\ve^4)\qquad 
\frac{1}{\widetilde a^2}=\frac{1}{a^2}\Big(\frac{1}{a}-\frac{\ln(e^{4\pi/a}-1)}{4\pi}\Big)\,,
\ee
\be
\lla 3=\frac{2}{a^2}\ve^2|\ln \ve|+o(\ve^2|\ln\ve|)\,.
\ee 
Since the leading order in the expansion does not change in the terms $\lla k$ with $k=4,5,...$, the estimate \eqref{th5Eveth0} for the eigenvalue of $H_\ve$ holds true.\\

We conclude the proof of the theorem with the analysis of the case $\theta_0<0$. As usual the eigenvalues of $H_\ve$ are given by the solutions of the equation $D_\ve(-\la)=0$, $\la>0$, which reads
\begin{equation}
\label{videotape}
\frac{\ve^2}{a}-(a|\theta_0|+\ve^2)\frac{\ln(1+\la)}{4\pi}=\frac{\ve^2}{a+|\theta_0|+\frac{\ve^2}{a}-\frac{(|\theta_0|a+\ve^2)\ln \la}{4\pi}}\qquad \la>0\,.
\end{equation}
Let us denote by $f^l(\la)$ and $f^r(\la)$ respectively the l.h.s. and the r.h.s. of the last equation. From the analysis of the functions $f^l(\la)$ and $f^r(\la)$ one can see that the equation \eqref{videotape} has two solutions. The function $f^r$ has a vertical asymptote in $\la_{a,\ve} =\exp\big[4\pi\big(|\theta_0|+a+\ve^2/a\big)/\big(|\theta_0|a+\ve^2\big)\big]$ and is strictly increasing in $(0,\la_{a,\ve})$ and $(\la_{a,\ve},\infty)$. Moreover it equals zero for $\la=0$. The function $f^l$ is strictly decreasing, it equals $\ve^2/a$ for $\la=0$ and zero in $\la=\lla0=-1+\exp\big(4\pi\ve^2/[a(a|\theta_0|+\ve^2)]\big)$. Then one solution is in $(\la_{a,\ve},\infty)$. This solution converges to the eigenvalue of the unperturbed Hamiltonian $H_0$, see  the proposition \ref{spectrumH0}. The other solution is in $(0, \lla 0)$. A sequence which converges to this solution  is defined by  $\lla{k+1}={f^l}^{-1}\big(f^r(\lla k)\big)$, $k=0,1,2,...$, where we denoted by  ${f^l}^{-1}$ the inverse function of $f^l$. We do not discuss the convergence of the sequence $\{\lla k\}$ which can be deduced also by the analysis of the plot of the functions $f^l$ and $f^r$. The estimate \eqref{th5Evethpos} can be obtained by a direct computation of the term $\lla 1$ and noticing that for any $k=1,2,3,...$ the correction to the term of order $\ve^2$ in $\lla k$  is of order $\ve^2/|\ln\ve|$.
\end{proof}

We finally consider  the case $d=3$. 
%%%%%%%
%THEOREM
%%%%%%%
\begin{theorem}
Let $d=3$ and assume that $c=0$. Then for all $\ve>0$ the  essential  spectrum of $H_\ve$  fills the positive real line and  is only absolutely continuous. Moreover the Hamiltonian $H_\ve$ has a zero energy resonance, i.e.,  the resolvent $R_\ve(z)$ has a singularity of order $z^{-1/2}$  in $z=0$.
\end{theorem}
%%%%%%%
%REMARK
%%%%%%%
\begin{remark}
For $d=3$, as opposed to the cases $c>0$ and $c<0$, when $c=0$ the singularity of the resolvent $R_0(z)$ in $z=0$ does not move from the origin. The perturbation parameter $\ve$ affects only the character of the singularity, turning the embedded eigenvalue into a zero energy resonance.

Similarly to what happens in $d = 1$, when $\theta_0 = 0$, there are two solutions of $D_\ve(z)=0$ in a neighborhood of the origin.  Both of them are in this case on the negative real axis of the second Riemann sheet of $\sqrt{z}$.

We also notice that for $\theta_0<0$ the Hamiltonian $H_\ve$ has an isolated eigenvalue in a neighborhood of order $\ve^2$ of the point $-(4\pi)^2/\theta_0^2$.
\end{remark}
%%%%%%
%PROOF
%%%%%%
\begin{proof}
For  $c=0$ the equation $D_\ve(z)=0$ reads
\be
D_\ve(z)=\bigg[1-i\frac{\theta_0}{4\pi}\sqrt{z}\bigg]\bigg[1+i\sqrt{z-1}\bigg]+\left(\frac{\ve}{4\pi}\right)^2\sqrt{z-1}\sqrt{z}=0\,.
\ee
By a direct analysis one see that the last equation has no solutions on the real positive axes, $z=\la>0$. It is easy to verify that $z=0$ is a solution of the last equation for all $\ve>0$. More precisely one can see that the function $D_\ve(z)$ can be expanded around $z=0$ as $
D_\ve(z)=i(\ve^2/(4\pi)^2)\sqrt{z}+\OO(|z|)$. Correspondingly, for all $\ve>0$, the   resolvent $R_\ve(z)$  has the following expansion around $z=0$
\begin{equation}
\label{exp}
R_\ve(z)=\frac{A_\ve}{\sqrt{z}}+B_\ve+\OO(|z|)
\end{equation}
where $A_\ve$ is the matrix valued operator with integral kernel 
\be
A_\ve(x',x)=
\begin{pmatrix}
\displaystyle\frac{i}{4\pi}\frac{1}{|x|}\frac{1}{|x'|}&\displaystyle-\frac{i}{\ve}\frac{1}{|x|}\frac{e^{-|x'|}}{|x'|}\\ \\
\displaystyle-\frac{i}{\ve}\frac{e^{-|x|}}{|x|}\frac{1}{|x'|}&\displaystyle\frac{4\pi i}{\ve^2}\frac{e^{-|x|}}{|x|}\frac{e^{-|x'|}}{|x'|}\\
\end{pmatrix}
\ee
According to standard results on the low energy expansion of resolvents of Schr\"odinger operators   in dimension three (see, e.g., \cite{jensen-kato:79} and \cite{jensen-nenciu:01}), the presence of a singularity of order $1/2$ is the signature of  a zero energy resonance. In the following we give the explicit form of the resonant state without making use of expansion \eqref{exp}. For this reason, we will not specify any suitable topology in order to give equality  \eqref{exp} a rigorous meaning.

 Let us show that, for all $\ve\neq0$, there exists a distributional solution of the equation $H_\ve\Phi^r=0$. Consider the state $\Phi^r$  given by
\[
\Phi^{r}=\begin{pmatrix}
\displaystyle\phi_0^{r} \\ \\ 
\displaystyle
\phi_1^{r}
\end{pmatrix}=N\begin{pmatrix}
\displaystyle -\frac{\ve}{4\pi|\cdot|} \\ \\
\displaystyle
\frac{e^{-|\cdot|}}{|\cdot|}
\end{pmatrix}
\]
where $N$ is an inessential multiplicative constant which we set equal to 1. 

Let us verify that $\Phi^{r}$ is a zero energy resonance for $H_\ve$. The function $\phi_0^{r}\in L^2_{loc}(\RE^3)$ but  $\phi_0^{r}\notin L^2 (\RE^3)$, then $\Phi^{r}\notin D(\hh_\ve)$. Nevertheless  using formulas given in  definition  \ref{def:Hve},  it  is possible  to  compute the charges $q^{r}_0$ and $q^{r}_1$ and the regular parts $f_0^{r}$ and $f_1^{r}$ associated to $\Phi^{r}$. A simple calculation gives $q^{r}_0=-\ve$, $q^{r}_1=4\pi$, $f_0^{r}=0$ and $f_1^{r}=-1$. Since the conditions $q^{r}_0=\ve f_1^{r}$ and $q^{r}_1=-4\pi  f_1^{r}$ are satisfied, one has that  for all $\Psi\in C_0^{\infty}(\RE^3)\oplus C_0^{\infty}(\RE^3)$ the equation  $(\Psi,H_\ve\Phi^{r})_\HH=0$ is satisfied.
\end{proof}

%%%%%%%%%%%%%%%%%%%%%%%%%
%SUBSECTION
%%%%%%%%%%%%%%%%%%%%%%%%%
\subsection{Negative perturbations. $c>0$.}

In this section we study the spectral structure of the Hamiltonian $H_\ve$ in the vicinity of the origin, when the parameter $c$ is positive. It is clear form \eqref{b0d1},  \eqref{b0d2}, \eqref{b0d3}, that this  choice corresponds to perturbations of the Hamiltonian $H_0$  for which the  threshold eigenvalue is pushed toward negative energies  by the perturbative term  $c\ve$ in $\theta_1^\ve$.

%%%%%%%
%THEOREM
%%%%%%%
\begin{theorem}
\label{th7}
Let $d=1,2,3$ and assume that $c>0$. Then  there exists $\ve_0>0$ such that for all $0<\ve<\ve_0$:\\
the  essential  spectrum of $H_\ve$   fills the positive real line and  is only absolutely continuous;\\
the Hamiltonian $H_\ve$ has an isolated eigenvalue in $z=E_\ve$ and 
\begin{align}
&E_\ve=
-c\ve+\OO(\ve^{3/2})&& d=1
\label{prop3d1}
\\
&E_\ve=-\frac{4\pi c}{a^2}\ve+\OO(\ve^2|\ln\ve|)&&d=2
\label{prop3d2}
\\
&E_\ve=-\frac{c\ve}{2\pi}+\OO(\ve^2)&&d=3\,.
\label{prop3d3}
\end{align}
\end{theorem}
%%%%%%%
%REMARK
%%%%%%%
\begin{remark}
Notice that when $c>0$ the eigenvalue in the upper channel moves toward negative energies and is smoothly perturbed by the channel coupling (compare equations \eqref{prop3d1} -  \eqref{prop3d3} with equations \eqref{b0d1} - \eqref{b0d3}). 

\noindent
For $d=1,2$,  when $\theta_0=0$, the analytic continuation of $D_\ve(z)$ from the semi-plane $\Im {z} >0$ through the positive real axis has zeros close to the origin.

As in the case $c =0$, see remark \ref{photo}, in dimension one, one of the roots of $D_\ve(z) =0$ has positive real part. The channel coupling produces this resonance as perturbation of the zero energy resonance in the lower channel.

\noindent
For $d=2$ the Hamiltonian $H_\ve$ has also one eigenvalue which moves to minus infinity as $\ve$ goes to zero.
\end{remark}

The order in $\ve$ of the remainder in expansions \eqref{prop3d1} and \eqref{prop3d2}  are exact for $\theta_0=0$. It is possible to prove that  for $\theta_0\neq0$ the remainders in equations  \eqref{prop3d1} and \eqref{prop3d2}  are indeed  $\OO(\ve^2)$.

In all dimensions the proof of theorem \ref{th7} comes straightly  from the analysis of the equation $D_{\ve}(-\la)=0$, $\la>0$, and we omit details.

%%%%%%%%%%%%%%%%%%%%%%%%%%%%%%%%%%%%%%%
%%%%%%%%%%%%%%%%%%%%%%%%%%%%%%%%%%%%%%%
%SECTION 
%%%%%%%%%%%%%%%%%%%%%%%%%%%%%%%%%%%%%%%
%%%%%%%%%%%%%%%%%%%%%%%%%%%%%%%%%%%%%%%
\section{
\label{sec3}
Conclusions}

We investigated the spectral properties of model-Hamiltonians describing a quantum particle interacting with a localized spin via zero range forces. Parameters were adjusted in such a way that the unperturbed Hamiltonian showed spectral singularities at the continuum threshold (chosen to be the zero energy point). 

\n
All the unperturbed Hamiltonians we considered had a zero energy eigenvalue. In addition, for particular values of parameters, a zero energy resonance was also present. Our aim has been to characterize the effect of perturbations on the spectral structure around the threshold.

\n
We defined perturbed Hamiltonians introducing a coupling between the two channels, associated to the two possible values of one component of the spin, together with potential-like perturbations of the upper channel unperturbed Hamiltonian. 

\n
Direct, sometimes lengthy, calculations bring us to results which agrees with the ones in \cite{jensen-nenciu:06} for positive perturbations. In fact the family of Hamiltonians we define have resolvents which are finite rank perturbations of the unperturbed resolvent. Roughly, the great part of the considerable work done in \cite{jensen-nenciu:06}, in a quite general setting, was to prove asymptotic expansions of the resolvent where only finite rank operators appear. This property holds true by construction for all our Hamiltonians reducing significantly analytic difficulties and enhancing explicit computability.

\n
The simplification mentioned above enables us to investigate also purely off-diagonal perturbations, where the perturbing term has no explicit bias to move singularities toward larger or lower values of energy. Moreover we prove that all the expansions for the singularity coordinates  are convergent and we give easy recurrent procedure to compute each term in the expansions.

\n
In particular we want to mention the one and two dimensional pure off-diagonal cases ($c =0$ and $\theta_0 =0$ in our notation), when two singularities are present in the unperturbed Hamiltonian spectrum. Our results show a peculiar spectral structure  of the corresponding Hamiltonian. As expected, no continuity of the spectral properties in parameter space is observed in this particular point.

\n
In this paper we analyzed only multi-channel Hamiltonians. As we mentioned in the introduction, one-channel Schr\"odinger operators describing a quantum particle interacting with many point scattering centers show very rich spectral structures and can suitably approximate Hamiltonians with any kind of smooth potentials. Moreover their resolvents are  finite rank perturbations of the Laplacian resolvent for any (finite) number of scattering centers.  In our opinion such kind of Hamiltonians are good candidates to examine spectral properties of a vast class of Schr\"odinger operators.

\vspace{0.3cm}

{\bf Acknowledgments} This work started when two authors, C.C. and R.C.,  were employed at the Doppler Institute (Czech Republic) and was partially supported by the institute grant (LC06002).\\

\vspace{0.3cm}

%%CLAUDIO
%%BIBLIOGRAFIA SUL MIO POWERBOOK
%\bibliographystyle{/Users/claudio/works/myamsplain}
%\bibliography{/Users/claudio/works/mywwb10}
%
%%CLAUDIO
%%BIBLIOGRAFIA SUL MAC DELL'UNIVERSITA'
%\bibliographystyle{/automount/Servers/him-addc-01.him.uni-bonn.de/home/cacciapuoti/works/myamsplain}
%\bibliography{/automount/Servers/him-addc-01.him.uni-bonn.de/home/cacciapuoti/works/mywwb10}

\begin{thebibliography}{10}

\bibitem{abramowitz-stegun:72}
M.~Abramowitz and I.~A. Stegun, \emph{Handbook of mathematical functions with
  formulas, graphs, and mathematical tables}, Dover Publications, Inc., New
  York, 1992.

\bibitem{aghh:05}
S.~Albeverio, F.~Gesztesy, R.~{H{\o}egh-Krohn}, and H.~Holden, \emph{Solvable
  models in quantum mechanics: {S}econd edition}, AMS Chelsea Publ., 2005, with
  an Appendix by P. Exner.

\bibitem{cacciapuoti-carlone-figari:07}
C.~Cacciapuoti, R.~Carlone, and R.~Figari, \emph{Spin dependent point
  potentials in one and three dimensions}, J. Phys. A: Math. Theor. \textbf{40}
  (2007), 249--261.

\bibitem{cacciapuoti-carlone-figari:09}
C.~Cacciapuoti, R.~Carlone, and R.~Figari, \emph{Resonances in models of spin
  dependent point interactions}, J. Phys. A: Math. Theor. \textbf{42} (2009),
  035202.

\bibitem{DO88}
Y.~N. Demkov and V.~N. Ostrovskii, \emph{Zero-range potentials and their
  applications in atomic physics}, Plenum Press, New York, 1988.

\bibitem{DES95}
P.~Duclos, P.~Exner, and P.~{\v{S}}t'ov\'{i}\v{c}ek, \emph{Curvature-induced
  resonances in a two-dimensional {D}irichlet tube}, Ann. Inst. H. Poincar\'e
  (A) Phys. Th\'{e}o. \textbf{62} (1995), 81--101.

\bibitem{DEM01}
P.~Duclos, P.~Exner, and B.~Meller, \emph{Open quantum dots: resonances from
  perturbed symmetry and bound states in strong magnetic fields}, Rep. Math.
  Phys. \textbf{47} (2001), 253--267.

\bibitem{Exn91}
P.~Exner, \emph{A solvable model of two-channel scattering}, Helv. Phys. Acta
  \textbf{64} (1991), 592--609.

\bibitem{Hun90}
W.~Hunziker, \emph{Resonances, metastable states and exponential decay laws in
  perturbation theory}, Comm. Math. Phys \textbf{132} (1990), 177--188.

\bibitem{jensen-nenciu:06}
A.~Jensen and G.~Nenciu, \emph{The {F}ermi golden rule and its form at
  thresholds in odd domensions}, Comm. Math. Phys. \textbf{261} (2006),
  693--727.

\bibitem{jensen-kato:79}
A.~Jensen and T.~Kato, \emph{Spectral properties of {S}chr{\"o}dinger operators
  and time decay of the wave functions}, Duke. Math. Jour. \textbf{46} (1979),
  583--611.

\bibitem{jensen-nenciu:01}
A.~Jensen and G.~Nenciu, \emph{A unified approach to resolvent expansions at
  thresholds}, Rev. Math. Phys. \textbf{13} (2001), no.~6, 717--754.

\bibitem{Kin91}
C.~King, \emph{Exponential decay near resonance, without analyticity}, Lett.
  Math. Phys. \textbf{23} (1991), 215--222.

\bibitem{SW98}
A.~Soffer and M.~I. Weinstein, \emph{Time dependent resonance theory}, Geom.
  Funct. Anal. \textbf{8} (1998), 1086--1128.

\end{thebibliography}

%%%%%%%%%%%%%%%%%%%%%%%%%%%%%%%%%%%%%%%%
%%%%%%%%%%%%%%%%%%%%%%%%%%%%%%%%%%%%%%%%
%%%%%%%%%%%%%%%%%%%%%%%%%%%%%%%%%%%%%%%%
%%%%%%%%%%%%%%%%%%%%%%%%%%%%%%%%%%%%%%%%
%FINE
%%%%%%%%%%%%%%%%%%%%%%%%%%%%%%%%%%%%%%%%
%%%%%%%%%%%%%%%%%%%%%%%%%%%%%%%%%%%%%%%%
%%%%%%%%%%%%%%%%%%%%%%%%%%%%%%%%%%%%%%%%
%%%%%%%%%%%%%%%%%%%%%%%%%%%%%%%%%%%%%%%%

\end{document}